\documentclass[runningheads]{llncs}
\usepackage{graphicx}
\usepackage{hyperref}

\newcommand{\no}[1]{}

\usepackage{amsfonts}
\usepackage{float}
\usepackage{tikz}
\usepackage{xcolor}
\usetikzlibrary{automata, positioning, arrows}
\tikzset{
->, 
>=stealth', 
node distance=2.5cm, 
every state/.style={fill=gray!10, minimum size=1cm}, 
initial text=$ $ 
}
\newcommand{\set}[2]{\{#1\,|\,#2\}}
\usepackage{amsmath}

\begin{document}
\title{On Stricter Reachable Repetitiveness Measures \thanks{Funded in part by Basal Funds FB0001, Fondecyt Grant 1-200038, and a Conicyt Doctoral Scholarship, ANID, Chile.}}
%
%
\author{Gonzalo Navarro \and
Cristian Urbina}
\authorrunning{Navarro and Urbina}
%
\institute{CeBiB --- Center for Biotechnology and Bioengineering \\ Departament of Computer Science, University of Chile}

\maketitle              
\begin{abstract}
The size $b$ of the smallest bidirectional macro scheme, which is arguably the most general copy-paste scheme to generate a given sequence, is considered to be the strictest reachable measure of repetitiveness. It is strictly lower-bounded by measures like $\gamma$ and $\delta$, which are known or believed to be unreachable and to capture the entropy of repetitiveness. 
In this paper we study another sequence generation mechanism, namely compositions of a morphism. We show that these form another plausible mechanism to characterize repetitive sequences and define NU-systems, which combine such a mechanism with macro schemes. We show that the size $\nu \le b$ of the smallest NU-system is reachable and can be $o(\delta)$ for some string families, thereby implying that the limit of compressibility of repetitive sequences can be even smaller than previously thought. We also derive several other results characterizing $\nu$.

\end{abstract}
\section{Introduction}

The study of repetitiveness measures, and of suitable measures of compressibility of repetitive sequences, has recently attracted interest thanks to the surge of repetitive text collections in areas like Bioinformatics, and versioned software and document collections. A recent survey \cite{Navarro2021} identifies a number of those measures, separating those that are {\em reachable} (i.e., any sequence can be represented within that space) from those that are not, which are still useful as lower bounds.

Reachable measures are, for example, the size $g$ of the smallest context-free grammar that generates the sequence \cite{KY00}, the size $c$ of the smallest {\em collage system} that generates the sequence \cite{Kida2003} (which generalizes grammars), the number $z$ of phrases of the Lempel-Ziv parse of the sequence \cite{LZ76}, or the number $b$ of phrases of a {\em bidirectional macro scheme} that represents the sequence \cite{Storer1982}. Such a macro scheme cuts the sequence into phrases so that each phrase either is an explicit symbol or it can be copied from elsewhere in the sequence, in a way that no cyclic dependencies are introduced. As such, macro schemes are the ultimate measure of what can be obtained by ``copy-paste'' mechanisms, which characterize repetitive sequences well.

Other measures are designed as lower bounds on the compressibility of repetitive sequences: $\gamma$ is the size of the smallest {\em string attractor} for the sequence \cite{Kempa2018} and $\delta$ is a measure derived from the string complexity \cite{Raskhodnikova2013,CEKNP19}. 

In asymptotic terms, it holds $\delta \le \gamma \le b \le c \le z \le g$ and, except for $c \le z$, there are string families where each measure is asymptotically smaller than the next. The recent result by Bannai et al.~\cite{Bannai2021}, showing that there exists a string family where $\gamma = o(b)$, establishes a clear separation between unreachable lower bounds ($\delta$,$\gamma$) and reachable measures ($b$ and the larger ones).

Concretely, Bannai et al.\ show that $b=\Theta(\log n)$ and $\gamma=O(1)$ for the Thue-Morse family, defined as $t_0=0$ and $t_{k+1} = t_k \overline{t_k}$, where $\overline{t_k}$ is $t_k$ with $0$s converted to $1$s and vice versa. This family is a well-known example of the fixed point of a {\em morphism} $\phi$, defined in this case by the rules $0 \rightarrow 01$ and $1 \rightarrow 10$. Then, $t_k$ is simply $\phi^k(0)$. This representation of the words in the family is of size $O(1)$, and each word can be easily produced in optimal time by iterating the morphism.

Iterating a small morphism is arguably a relevant mechanism to define repetitive sequences. Intuitively, any short repetition $\alpha[1,k]$ that arises along the generation of a long string turns into a longer repetition $\phi^t(\alpha[1])\cdots\phi^t(\alpha[k])$ in the final string, $t$ steps later. More formally, if a morphism is $k$-uniform (i.e., all its rules are of fixed length $k$), then the resulting sequence is so-called {\em $k$-automatic} \cite{Shallit2003} and its prefixes have an attractor of size $\gamma = O(\log n)$ \cite{Shallit2020}. That is, many small morphisms lead to sequences with low measures of repetitiveness. Further, in the Thue-Morse family, morphisms lead to a reachable measure of repetitiveness that is $o(b)$, below what can be achieved with copy-paste mechanisms.

In this paper we further study this formalism. First, we define {\em macro systems}, a grammar-like extension that we prove equivalent to bidirectional macro schemes. We then study {\em deterministic Lindenmayer systems} \cite{Lindenmayer1968a,Lindenmayer1968b}, a grammar-like mechanism generating infinite strings via iterated morphisms; they are stop\-ped at some level to produce a finite string. We combine both systems into what we call {\em NU-systems}. The size $\nu$ (``nu'') of the smallest NU-system is always reachable and $O(b)$. Further, we show that there are string families where $\nu = o(\delta)$, thereby showing that $\delta$ is not anymore a lower bound for the compressibility of repetitive sequences if we include other plausible mechanisms to represent them. We present several other results that help characterize the new measure $\nu$.

\section{Basic Concepts}


\subsection{Terminology}
Let $\Sigma$ be a set of symbols, called the \textit{alphabet}. A string $w$ of length $|w| = n$ (also denoted as $w[1,n]$ when needed) is a concatenation of $n$ symbols from $\Sigma$; in particular the string of length $0$ is denoted by $\varepsilon$. The set of $k$-length concatenations of symbols from $\Sigma$ is denoted $\Sigma^k$, and the set $\Sigma^*$ of strings over $\Sigma$ is defined as $\bigcup_{k\ge 0}\Sigma^k$; we also define $\Sigma^+ = \bigcup_{k\ge 1}\Sigma^k$. We juxtapose strings ($xy$) or combine them with the dot operator ($x\cdot y$) to denote their concatenation.
A string $x$ is a \textit{prefix} of $w$ if $w = xz$, a {\em suffix} of $w$ if $w=yx$, and a {\em substring} of $w$ if $w=yxz$, for some $y,z \in \Sigma^*$. Let $w[1,n]$ denote an $n$-length string. Then $w[i]$ is the $i$-th symbol of $w$, and $w[i,j]$ the substring $w[i]w[i+1]\cdots w[j]$ if $1 \leq i \leq j \leq n$, and $\varepsilon$ if $j < i$.


\subsection{Parsing based schemes}

Probably the most popular measure of repetitiveness is the number $z$ of {\em phrases} in the so-called {\em Lempel-Ziv parse} of a word $w[1,n]$ \cite{LZ76}. In such a parse, $w$ is partitioned optimally into phrases $w = x_1 \cdots x_z$, so that every $x_k$ is either of length $1$ or it appears starting to the left in $w$ (so the phrase $x_k$ is copied from some {\em source} at its left). This parsing can be computed in $O(n)$ time.

Storer and Szymanski \cite{Storer1982} introduced {\em bidirectional macro schemes}, which allow sources appear to the left or to the right of their phrases, as long as circular dependencies are avoided. We follow the definition by Bannai et al.~\cite{Bannai2021}.

Let $w[1,n]$ be a string. A bidirectional macro scheme of size $k$ for $w$ is a sequence $B = (x_1, s_1),\ldots,(x_k, s_k)$ satisfying $w = x_1\cdots x_k$ and $x_i = w[s_i, s_i + |x_i| - 1]$ if $|x_i| > 1$, and $s_i = \bot$ if $|x_i| = 1$. We denote the starting position of $x_i$ in $w$ by $p_i = 1+\sum_{j=1}^{i-1} |x_j|$. The function $f: [1,n] \rightarrow [1,n] \cup \{\bot\}$,

\begin{equation*}
  f(i) = \left\{
     \begin{array}{@{}l@{\thinspace}l}
       \bot  &: \text{ if } s_i = \bot,\\
       s_j + i - p_j&:  \text{ if } p_j \leq s_i < p_{j+1}, \\
     \end{array}
   \right.
\end{equation*}
is induced by the macro scheme. For $B$ to be a valid bidirectional macro scheme it must hold that, for each $i$, there exists some $r$ satisfying $f^r(i) = \bot$. Therefore, it suffices with the values $|x_i|$ and $s_i$, plus $x_i$ where $s_i=\bot$, to recover $w$.

We call $b \le z$ the
number $k$ of elements in the smallest bidirectional macro
scheme generating a given string $w[1,n]$. There are string families where $b=o(z)$ \cite{NOP20}. While $z$ is computed in linear time, computing $b$ is NP-hard.

\subsection{Grammars and generalizations}

The size $g$ of the smallest context-free grammar generating (only) a word $w[1,n]$ \cite{KY00} is a relevant measure of repetitiveness. Such a grammar has exactly one rule per nonterminal, and those can be sorted so that the right-hand sides mention only terminals and previously listed nonterminals. The size of the grammar is the sum of the lengths of the right-hand sides of the rules. The {\em expansion} of a nonterminal is the string of terminals it generates; the word defined by the grammar is the expansion of its last listed nonterminal.

More formally, a grammar over the alphabet of terminals $\Sigma$ is a sequence of nonterminals $X_1,\ldots,X_r$, with a rule $X_k \rightarrow A_{k,1}\cdots A_{k,\ell_k}$, each $A_{k,r}$ being a terminal or some nonterminal in $X_1,\ldots,X_{k-1}$. The expansion of a terminal $a$ is $exp(a)=a$, and that of a nonterminal $X_k$ is $exp(X_k) = exp(A_{k,1})\cdots exp(A_{k,\ell_k})$. The grammar represents the string $w=exp(X_r)$, and its size is $\sum_{k=1}^r \ell_k$.

Composition systems were introduced by Gasieniec et al.~\cite{Gasieniec1996}. Those add the ability to reference any prefix or suffix of the expansion of a previous nonterminal (and, thus, substrings as prefixes of suffixes). Let us use the more general form, allowing terms $A_{k,r} = X_j[s,t]$ where $exp(X_j[s,t]) = exp(X_j)[s,t]$.

Kida et al.~\cite{Kida2003} extended composition systems with \textit{run-length} terms of the form $(A_{k,r})^t$, so that
$exp((A_{k,r})^t) = exp(A_{k,r})^t$, the expansion of $A_{k,r}$ concatenated $t$ times. They called this extension a {\em collage system}. We call $c \le g$ the smallest collage system generating a word $w[1,n]$, and it always holds $b = O(c)$\footnote{At least if the collage system is {\em internal}, that is, every $exp(X_k)$ appears in $w$.} and $c=O(z)$. There are string families where $b=o(c)$ \cite{NOP20}, and where $z=o(g)$. Computing $g$ (and, probably, $c$ too) is NP-hard. 

\subsection{Lower bounds}

Kempa and Prezza introduced the concept of {\em string attractor} \cite{Kempa2018}, which yields an abstract measure that lower-bounds all the previous reachable measures.

Let $w[1, n]$ be a string. A string attractor for $w$ is a set of positions $A \subseteq [1, n]$ where for every substring $w[i, j]$ there exists a copy $w[i',j']$ (i.e., $w[i, j] = w[i',j']$) and a position $k \in A$ with $i' \leq k \leq j'$. The measure $\gamma$ is defined as the cardinality of the smallest of such attractors for a given string $w[1,n]$, and it always holds that $\gamma = O(b)$. Further, a string family where $\gamma = o(b)$ exists \cite{Bannai2021}.

Kociumaka et al.~\cite{CEKNP19} used the {\em string complexity} of $w[1,n]$ to define a measure called $\delta$. Let $S_w(k)$ be the number of distinct substrings of length $k$ in $w$. Then $\delta = \max\set{S_w(k)/k}{1 \leq k \leq n}$. This measure is computed in $O(n)$ time and it always holds that $\delta = O(\gamma)$; there are string families where $\delta = o(\gamma)$ \cite{KNP20}. While $\delta$ is unreachable in some string families, any string can be represented in $O(\delta\log(n/\delta))$ space \cite{KNP20}. Measure $\delta$ has been proposed as a lower bound on the compressibility of repetitive strings, which we question in this paper.

\subsection{Morphisms over strings}

We explain some general concepts about morphisms acting over strings \cite{Shallit2003,Lothaire2002}.
A \textit{monoid} $(M, *, e)$ is a set with an associative operation $*$ and a neutral element $e \in M$ satisfying $a * e = e * a = a$ for every $a \in M$. We write $ab$ for $a*b$ and say that $M$ is a monoid, instead of $(M, *, e)$. A {\em morphism of monoids} is a function $\phi : M_1 \rightarrow M_2 $, where $(M_1,*_1,e_1)$ and $(M_2,*_2,e_2)$ are monoids, $\phi(a *_1 b) = \phi(a) *_2 \phi(b)$ for every $a,b \in M_1$, and $\phi(e_1) = e_2$.

Let $\Sigma$ be a set of symbols, and $\cdot$ the concatenation  of strings. Then $(\Sigma^*, \cdot,\varepsilon)$ is a monoid with string concatenation, called the \textit{free monoid}. A {\em morphism of free monoids} $\phi: \Sigma^* \rightarrow \Delta^*$ is defined completely just by specifying $\phi$ on the symbols on $\Sigma$. If $\Sigma$ = $\Delta$, then $\phi$, is called an \textit{automorphism}, and $\phi$ is {\em iterable}. We define the $n$-iteration (or composition) of $\phi$ over $s$ as $\phi^n(s)$.

Let $\phi: \Sigma^* \rightarrow \Delta^*$ be a morphism of free monoids. We define $depth(\phi) = |\Sigma|$, $width(\phi) = \max_{a\in \Sigma}|\phi(a)|$, and $size(\phi) = \sum_{a \in \Sigma}|\phi(a)|$. We say $\phi$ is {\em expanding} if $|\phi(a)| > 1$, {\em non-erasing} if $|\phi(a)| > 0$, and {\em $k$-uniform} if $|\phi(a)| = k$, for every $a \in \Sigma$. A {\em coding} is a 1-uniform morphism. We say $\phi$ is {\em prolongable} on $a\in \Sigma$ if $\phi(a) = as$ for a non-empty string $s$.

Let $\phi$ be an automorphism on $\Sigma^*$. Let $\phi$ be prolongable on $a$, so $\phi(a) = as$. Then, $w = as\phi(a)\phi^2(a)\cdots$ is the unique {\em fixed point} of $\phi$ starting with $a$, that is, $\phi(w) = w$ \cite{Lothaire2002}. Words constructed in this fashion are called {\em purely morphic words}. If we apply a coding to them, we obtain {\em morphic words}. A morphic word obtained from a $k$-uniform morphism is said to be $k${\em -automatic} \cite{Shallit2003}.

\section{Macro Systems}

Our first contribution is the definition of \textit{macro systems}, a generalization of composition systems we prove to be as powerful as bidirectional macro schemes. That is, the smallest macro system generating a given string $w$ is of size $O(b)$.

\begin{definition}
A {\em macro system} is a tuple $M = (V, \Sigma, R, S)$, where $V$ is a finite set of symbols called the \textit{variables}, $\Sigma$ is a finite set of symbols disjoint from $V$ called the \textit{terminals}, $R$ is the set of rules (exactly one per variable)
\begin{equation*}
R : V \rightarrow (V \cup \Sigma \cup 
\{A[i,j] ~|~ A \in V, i,j\in \mathbb{N} \})^*,\end{equation*}
and $S \in V$ is the \textit{initial variable}. If $R(A) = \alpha$ is the rule for $A$, we also write $A \rightarrow \alpha$. The symbols $A[i,j]$ are called {\em extractions}. The rule $A \rightarrow \varepsilon$ is permitted only for $A=S$. 
The {\em size} of a macro system is the sum of the lengths of the right-hand sides of the rules, $size(M)=\sum_{A \in V} |R(A)|$.
\end{definition}

We now define the string generated by a macro system as the expansion of its initial symbol, $exp(S)$. Such expansions are defined as follows.

\begin{definition}
Let $M=(V,\Sigma,R,S)$ be a macro system. The {\em expansion} of a symbol is a string over $\Sigma^*$ defined inductively as follows:
\begin{itemize}
\item If $a \in \Sigma$ then $exp(a) = a$.
\item If $S \rightarrow \varepsilon$, then $exp(S)=\varepsilon$.
\item If $A \rightarrow B_1 \cdots B_k$ is a rule, then
$exp(A) =exp(B_1) \cdots exp(B_k)$.
\item $exp(A[i,j]) = exp(A)[i,j]$ (this second $[i,j]$ denotes substring).
\end{itemize}

We say that the macro system is {\em valid} if there is a single solution $w \in \Sigma^*$ for $exp(S)$. We say that the macro system {\em generates} the string $w$.
\end{definition}

Note that a macro system looks very similar to a composition system, however, it does not impose an order so that each symbol references only previous ones. 
This algorithm determines the string generated by a macro system, if any:
\begin{enumerate}
    \item Compute $|exp(A)|$ for every nonterminal $A$, using the rules:
    \begin{itemize}
        \item If $a\in\Sigma$, then $|exp(a)|=1$.
        \item If $A \rightarrow B_1\cdots B_k$, then $|exp(A)|=|exp(B_1)|+\cdots+|exp(B_k)|$.
        \item $|exp(A[i,j])| = j-i+1$.
    \end{itemize}
    This must generate a system of equations without loops (otherwise the macro system is invalid), which is then trivially solved.
    \item Replace every symbol $A[i,j]$ by $A[i] \cdots A[j]$; we use $A[r]$ to denote $A[r,r]$.
    \item Replace every $A[r]$, if $A \rightarrow B_1\cdots B_k$, iterating until obtaining a terminal:
    \begin{itemize}
        \item Let $p_i = 1+\sum_{j=1}^{i-1} |exp(B_j)|$, for $1\le i \le k+1$.
        \item Let $s$ be such that $p_s \le r < p_{s+1}$.
        \item If $B_s \in \Sigma$, replace $A[r]$ by $B_s$.
        \item Otherwise replace $A[r]$ by $B_s[r-p_s+1]$.
    \end{itemize}
    \item If the process to replace any $A[r]$ falls in a loop (i.e., we return to $A[r]$), then the system has no unique solution and thus it is invalid. Otherwise, we are left with a classical context-free grammar without extractions, and compute $w =  exp(S)$ in the classical way.
\end{enumerate}

Note that a rule like $A \rightarrow B~ A[1,(t-1)|exp(B)|]$ solves only for $exp(A)=exp(B)^t$, just like the run-length symbol $B^t$ of collage systems. For example, $A \rightarrow ab$ and $S \rightarrow A~S[1,4]$ generates $ababab$ as follows:
\begin{eqnarray*}
& & A~S[1]~S[2]~S[3]~S[4] \\
& & A~A[1]~A[2]~S[1]~S[2] \\
& & A~a~b~A[1]~A[2] \\
& & A~a~b~a~b \\
& & a~b~a~b~a~b
\end{eqnarray*}
This shows that macro systems are at least as powerful as collage systems. But they can be asymptotically smaller. For example, the smallest collage system generating the Fibonacci string $F_k$ (where $F_1 = b$, $F_2 = a$, and $F_{k+2} = F_{k+1} F_k$) is of size $\Theta(\log |F_k|)$ \cite[Thm.~32]{NOP20}. Instead, we can mimic a bidirectional macro scheme of size $4$ \cite[Lem.~35]{NOP20} with a constant-sized macro system generating $F_k$: $S \rightarrow S[f_{k-2}+1,f_k-2]~b~a~S[f_{k-2}+1,2f_{k-2}]$ if $k$ is odd and 
$S \rightarrow S[f_{k-2}+1,f_k-2]~a~b~S[f_{k-2}+1,2f_{k-2}]$ if $k$ is even (where $f_k=|F_k|$). For example, for $F_7$ the system is $S \rightarrow S[6,11]~b~a~S[6,10]$ and we extract $F_7=exp(S)$ as follows, using that $F_7[1,6]=F_7[6,11]$, $F_7[7]=b$, $F_7[8]=a$, and $F_7[9,13]=F_7[6,10]$:
\begin{eqnarray*}
& & S[6]~S[7]~S[8]~S[9]~S[10]~S[11]~b~a~S[6]~S[7]~S[8]~S[9]~S[10] \\
& & S[11]~b~a~S[6]~b~a~b~a~S[11]~b~a~S[6]~b \\
& & a~b~a~S[11]~b~a~b~a~a~b~a~S[11]~b \\
& & a~b~a~a~b~a~b~a~a~b~a~a~b
\end{eqnarray*}

In general, we can prove that a restricted class of our macro systems is equivalent to bidirectional macro schemes.

\begin{definition}
A macro system $M=(V,\Sigma,R,S)$ generating $w$ is {\em internal} if $exp(A)$ appears in $w$ for every $A \in V$. We use $m$ to denote the size of the smallest internal macro system generating $w$.
\end{definition}

\begin{theorem} It always holds that $m \le b$. 
\end{theorem}

\begin{proof}
Let $(x_1,s_1),\ldots,(x_b, s_b)$ be the smallest bidirectional macro scheme generating $w[1,n] = x_1\cdots x_b$. We construct a macro system $M = (\{S\}, \Sigma, R, S)$ with a single rule $S \rightarrow A_1 \cdots A_b$, where
$A_i$ is the single terminal $x_i$ if $s_i=\bot$, and the extraction symbol
$S[s_i,s_i+|x_i|-1]$ if not.

We now show that this macro system is valid. 
After we execute step 2 of our algorithm, the length of the resulting string (which we call $W$) is already $n$: it has only terminals and symbols of the form $W[i] = S[r]$. Note that this implies that $f(i)=r$ in the bidirectional macro scheme. In every step, we replace each such $S[r]$ by $W[r]$. Since the macro scheme is valid, for each $i$ there is a finite $k$ such that $f^k(i)=\bot$, and thus $W[i]$ becomes a terminal symbol after $k$ steps. \qed
\end{proof}

\begin{theorem}
For every internal macro system of size $m$ there is a bidirectional macro scheme of size $b \le m$.
\end{theorem}
\begin{proof}
An internal macro system $M=(V,\Sigma,R,S)$ generating $w[1,n]$ can always be transformed into one with a single rule for the initial symbol. Let $A \in V$ be such that $W[i,j] = exp(A)$. We can then replace every occurrence of $A$ by $S[i,j]$, and every occurrence of $A[i',j']$ by $S[i'+i-1,j'+i-1]$, on the right-hand sides of all the rules. In particular, the rule defining $S$ will now contain terminals and symbols of the form $S[i,j]$, and thus all the other nonterminals can be deleted.

From the resulting macro system $S \rightarrow A_1 \cdots A_{m'}$, where $m' \le m$, we can derive a bidirectional macro scheme $(x_1,s_1),\ldots,(x_{m'},s_{m'})$, as follows: if $A_t$ is a terminal, then $x_t$ is that terminal and $s_t=\bot$. Otherwise, $A_t$ is of the form $S[i,j]$ and then $x_t=w[i,j]$ and $s_t=i$. The resulting scheme is valid, because our algorithm extracts any $S[i]$ after a finite number $k$ of steps, which is then the $k$ such that $f^{k+1}(i)=\bot$. \qed
\end{proof}

That is, bidirectional macro schemes are equivalent to internal macro systems. General macro systems can be asymptotically smaller in principle, though we have not found an example where this happens.

\section{Deterministic Lindenmayer Systems}

In this section we study a mechanism for generating infinite sequences called {\em deterministic Lindenmayer Systems} \cite{Lindenmayer1968a,Lindenmayer1968b}, which build on morphisms. We adapt those systems to generate finite repetitive strings. Those systems are, in essence, grammars with only nonterminals, which typically generate longer and longer strings, in a levelwise fashion. For our purposes, we will also specify at which level $d$ to stop the generation process and the length $n$ of the string $w$ to generate. The generated string $w[1,n]$ is then the $n$-length prefix of the sequence of nonterminals obtained at level $d$. 
We adapt, in particular, the variant called CD0L-systems, though we will use the generic name {\em L-systems} for simplicity.

\begin{definition}
An {\em L-system} is a tuple $L=(V,R,S,\tau,d,n)$, where $V$ is a finite set of symbols called {\em variables}, $R : V \rightarrow V^+$ is the set of {\em rules}, $S \in V^*$ is a sequence of variables called the {\em axiom}, $\tau : V \rightarrow V$ is a coding, $d \in \mathbb{N}$ is the level where to stop, and $n \in \mathbb{N}$ is the length of the string to generate.

An L-system produces {\em levels} of strings $L_i \in V^*$, starting from $L_0=S$ at level 0. Each level replaces every variable $A$ from the previous level by $R(A)$, that is, $L_{i+1}=R(L_i)$ if we identify $R$ with its homomorphic extension. The generated string is $w[1,n] = \tau(L_d[1,n]) \in V^*$, seeing $\tau$ as its homomorphic extension.

The {\em size} of an L-system is $|S|+\sum_{A \in V} |R(A)|$. We call $\ell$ the size of the smallest L-system generating a string $w$.
\end{definition}

L-systems then represent strings by iterating a non-erasing automorphism. 
Somewhat surprisingly, we now exhibit a string family where $\delta = \Omega(\ell\log n)$, thus L-systems are a reachable mechanism to generate strings that can be asymptotically smaller than what was considered to be a stable lower bound.

\begin{theorem} There exist string families where $\delta = \Omega(\ell\log n)$.
\end{theorem}

\begin{proof} 
Consider the L-system $L=(V,R,S,\tau,d,n)$ where $V=\{0,1\}$, $S = 0$, $R(0) = 001$, $R(1)=1$, $\tau(0)=0$, $\tau(1)=1$, and $n=2^{d+1}-1$. The family of strings is formed by all those generated by the systems $L$, where $d \in \mathbb{N}$. It is clear that all the strings in this family share the value $\ell=5$.

The first strings of the family generated by this system (i.e., its levels $L_i$) are
$0$, $001$, $0010011$, $001001100100111$, and so on. It is easy to see by induction that level $i$ contains $2^i$ $0$s and $2^i-1$ $1$s, so the string $L_i$ is of length $2^{i+1}-1$. 

More importantly, one can see by induction that levels $i\ge 2$ start with $00$ and contain all the strings of the form $01^j0$ for $1 \le j < i$. This is true for level $2$. Then, in level $i+1$ the strings $01^j0$ become $0011^j001$, which  contains $01^{j+1}0$, and the first $00$ yields $001001$, containing $010$.

Consider now the number of $d$-length distinct substrings in $L_d$, for $d\ge 4$. Each distinct substring $01^j0$, for $\lfloor d/2\rfloor-1 \le j \le d-2$, yields at least $d-j-1$ distinct $d$-length substrings (containing $01^j0$ at different offsets; no single $d$-length substring may contain two of those).
These add up to $d^2/8 + d/4$ distinct $d$-length substrings, and thus $\delta=\Omega(d)=\Omega(\log n)$ on the string $w=\tau(L_d)$.\qed
\end{proof}

On the other hand, $L$-systems are always reachable, which yields the immediate result that $\delta$ and $\ell$ are incomparable.

\begin{theorem}
There exist string families where $\ell = \Omega(\delta\log n)$.
\end{theorem}
\begin{proof}
Kociumaka et al.~\cite[Thm.~2]{KNP20} exhibit a string family of $2^{\Theta(\log^2 n)}$ elements with $\delta=O(1)$, so it needs $\Omega(\log^2 n)$ bits, that is, $\Omega(\log n) = \Omega(\delta \log n)$ space, to be represented with any method. Therefore $\ell = \Omega(\log n) = \Omega(\delta\log n)$ in this family, because there are only $2^{O(\ell\log n)}$ distinct L-systems of size $\ell$.
\qed
\end{proof}

Those strings are formed by $n$ $a$s, replacing
them by $b$s at single arbitrary positions between $2 \cdot 4^{j-2}+1$ and $4^{j-1}$ for every $j \ge 2$. While such a string is easily generated by a composition system of size $\Theta(\log n)$, we could only produce L-systems of size $\Theta(\log^2 n)$ generating it. We now prove bounds between L-systems and context-free grammars.

\begin{theorem} 
For any L-system $L=(V,R,S,\tau,d,n)$
of size $\ell$ generating $w$, there is a context-free grammar of size $(d+1)\ell$ generating $w$. If the morphism represented by $R$ is expanding, then the grammar is of size $O(\ell\log n)$.
\end{theorem}
\begin{proof}
Consider the derivation tree for $w$ in $L$: the root children are $S=L_0$ at level 0, and if $A$ is a node at level $i$, then the children of $A$ are the elements in $R(A)$, at level $i+1$. The nodes in each level $i$ spell out $L_i$.

We create a grammar $G=(V',V,R',S')$ where $V'$ contains the initial symbol $S'$ and, for each variable $A \in V$ of the L-system, $d$ nonterminals $A_0,\ldots,A_{d-1}$. The terminals of the grammar are the set of L-system variables, $V$. Then, for each L-system rule $A \rightarrow B_1\cdots B_k$ appearing in level $0 \le i \le d-2$, we add the grammar rule $A_i \rightarrow  (B_1)_{i+1}\cdots (B_k)_{i+1}$. Further, for each rule $A \rightarrow B_1\cdots B_k$ appearing in level $d-1$, we add the grammar rule $A_{d-1} \rightarrow \tau(B_1)\cdots \tau(B_k)$. Finally, if $S = B_1\cdots B_k$ is the L-system axiom, we add the grammar rule $S' \rightarrow (B_1)_0 \cdots (B_k)_0$ for its initial symbol.

It is clear that the grammar is of size at most $(d+1)\ell$ and it generates $w$. If every rule is of size larger than $1$, and $d > \lg n$, then the prefix $w[1,n]$ of $\tau(L_d)$ is generated from the first symbol of $L_{d-\lceil \lg n \rceil}$, which can then be made the axiom and $d$ reduced to $\lceil \lg n \rceil$. In this case, the grammar is of size $O(\ell \log n)$.\qed
\end{proof}

For example, consider our L-system $0 \rightarrow 001$ and $1 \rightarrow 1$.
A grammar simulating a generation of $d=3$ levels contains the rules $S' \rightarrow 0_0$, 
$0_0 \rightarrow 0_10_11_1$,
$0_1 \rightarrow 0_20_21_2$, $1_1 \rightarrow 1_2$,
$0_2 \rightarrow 001$, and $1_2 \rightarrow 1$. Note how the grammar uses the level subindices to control the point where the L-system should stop.

On the other hand, while we believe that composition systems can be smaller than L-systems, we can prove that L-systems are not larger than grammars.

\begin{theorem}
It always holds that $\ell=O(g)$.
\end{theorem}
\begin{proof}
Consider a grammar $G=(V,\Sigma,R,S)$ of height $h$ generating $w[1,n]$. We define the L-system $L=(V\cup\Sigma,R',R(S),\tau,h,n)$, where $R'$ contains all the rules in $R$ except the one for $S$. We also include in $R$ the rules $a \rightarrow a$ for all $a \in \Sigma$. The coding $\tau$ is the identity function.

It is clear that this L-system produces the same derivation tree of $G$, reaching terminals $a$ at some level. Those remain intact up to the last level, $h$, thanks to the rules $a \rightarrow a$. At this point the L-system has derived $w[1,n]$. 

The size of the L-system is that of $G$ plus $|\Sigma|$, which is of the same order because every symbol $a \in \Sigma$ appears on some right-hand side (if not, we do not need to create the rule $a \rightarrow a$ for that symbol).
\qed
\end{proof}

The following simple result characterizes a class of morphisms generating families with constant-sized L-systems.

\begin{theorem}\label{thm:crescent_constant}
Let $w \in \Delta^*$, $\psi: \Delta^* \rightarrow \Delta^*$ be a non-erasing automorphism over free monoids, and $\tau: \Delta \rightarrow \Sigma$.
Then $\ell = O(1)$ on the family $\{\tau(\psi^d(w))\, |\, d > 0\}$.
\end{theorem}

\begin{proof} 
We can easily simulate $\psi$ on the L-system $L = (\Delta, R, w, \tau, d, n)$ of fixed size, with $R(a)=\psi(a)$ and $n=|\psi^d(w)|$. The system generates $\tau(\psi^d(w))$ and, as $d$ grows, it does not change its size. \qed
\end{proof} 

This implies that $\ell = O(1)$ on families of $n$-iterations of the Thue-Morse morphism, the Fibonacci morphism, images of $k$-uniform morphisms (i.e., morphisms generating $k$-automatic words \cite{Shallit2003}), and standard Sturmian morphisms \cite{deLuca1997}. More generally, $\ell$ is $O(1)$ on the set of prefixes of any morphic word.

\section{NU-Systems}

We now define a mechanism that combines both macro systems and L-systems, yielding a computable measure that is reachable and strictly better than $b$.

\begin{definition}
A {\em NU-system} is a tuple $N=(V,R,S,\tau,d,n)$, which is understood in the same way as L-systems, except that we extend rules with {\em extractions}, that is, $R : V \rightarrow (V \cup E)^+$ and
$$ E = \{A(l)[i,j] ~|~ A \in V, l,i,j\in \mathbb{N} \}.$$
The symbol $A(l)[i,j]$ means to expand variable $A$ for $l$ levels and then extract $\tau(A_l[i,j])$ from the string $A_l$ at level $l$, recursively expanding extractions if necessary. This counts as a single expansion (one level) of a rule, that is, the levels $L_i$ in the NU-system belong to $V^*$.
We also use $A(l) = A(l)[1,|A_l|]$ to denote the whole level $l$ of $A$. 
%
%
The {\em size} of the NU-system is $size(N)=|S|+\sum_{A \in V} |R(A)|$. We call $\nu$ the size of the smallest NU-system generating a string $w[1,n]$.
\end{definition}

Just as macro systems, a NU-system is {\em valid} only if it does not introduce circular dependencies. 
Let $maxl$ be the maximum $l$ value across every rule $A(l)[i,j]$ in the NU-system.
The following algorithm determines the string generated by the system, if any:
\begin{enumerate}
    \item Compute $|A_l|$ for every variable $A$ and level $0 \le l \le maxl$, using the rules:
    \begin{itemize}
        \item $|A_0|=1$.
        \item If $l > 0$ and $A \rightarrow B_1\cdots B_k$, then $|A_l|=|(B_1)_{l-1}|+\cdots+|(B_k)_{l-1}|$.
        \item Replace $|B(l)[i,j]| = j-i+1$ on the previous summands $|(B_r)_{l-1}|$.
    \end{itemize}
    This generates a system of equations without loops, which is trivially solved.
    \item Replace every symbol $A(l)[i,j]$ in $R$ by $A(l)[i] \cdots A(l)[j]$; we use $A(l)[r]$ to denote $A(l)[r,r]$.
    \item Expand the rules, starting from the axiom, level by level as in L-systems. Handle the symbols $A(l)[r]$ as follows:
    \begin{enumerate}
    \item Replace every $A(0)[r]$ (so $r=1$ if the NU-system is correct) by $\tau(A)$.
    \item Replace every $A(l)[r]$, if $l>0$ and $A \rightarrow B_1\cdots B_k$, as follows:
    \begin{itemize}
        \item Let $p_i = 1+\sum_{j=1}^{i-1} |(B_j)_{l-1}|$, for $1\le i \le k+1$.
        \item Let $s$ be such that $p_s \le r < p_{s+1}$.
        \item Replace $A(l)[r]$ by $B_s(l-1)[r-p_s+1]$.
    \end{itemize}
    \item Return to (a) until the extraction symbol disappears.
    \end{enumerate}
\end{enumerate}

Note that the symbol $B_s$ in step 3(b) can in turn be of the form $B_s=B(l')[r']$; we must then extract $B(l')[r']$ before continuing the extraction of $B_s(l-1)[r-p_s+1]$. If, along the expansion, we return again to the original $A(l)[r]$, then the system has no unique solution and thus it is invalid. This is computable because the number of possible combinations $A(l)[r]$ is bounded by $|V|\cdot maxl \cdot n$. 

We now show that NU-systems are at least as powerful as macro systems and L-systems.

\begin{theorem} It always holds that $\nu = O(\min(\ell,m))$.
\end{theorem}

\begin{proof}
It always holds $\nu \le \ell$ because L-systems are a particular case of NU-systems. With respect to $m$, let
$M = (V, \Sigma, R, S)$ be a minimal macro system generating $w[1,n]$. Then we construct a NU-System $N=(V\cup\Sigma, R', S,\tau,d,n)$ where $\tau$ is the identity and $d=|V|$, which upper-bounds the height of the derivation tree. Each level of $N$ will simulate the sequence of extractions that lead from each $A[r]$ to its corresponding terminal in the macro system.

For each $a \in \Sigma$ we define the rule $a \rightarrow a$ in $R'$. For each rule $A \rightarrow B_1 \cdots B_m$ in $R$, we define the rule $A \rightarrow B_1' \cdots B_m'$ in $R'$, where $B_i' = B_i$ if $B_i \in V \cup \Sigma$, and $B_i' = A'(d)[j,k]$ if $B_i = A'[j,k]$. 
It is not hard to see that the NU-System $N$ simulates the macro system $M$, and its size is $O(m)$.
\qed
\end{proof}

For example, consider our previous macro system $A \rightarrow ab$ and $S \rightarrow A ~ S[1,4]$. The corresponding NU-system would have the rules $a \rightarrow a$, $b \rightarrow b$, $A \rightarrow ab$, and $S \rightarrow A ~ S(2)[1,4]$. The derivation is then generated as follows:
\begin{center}
\begin{tabular}{rclcl}
$L_0$ & $=$ & $S$ & $\longrightarrow$ &
     $A ~ S(2)[1] ~ S(2)[2] ~ S(2)[3] ~ S(2)[4]$ \\
& & & & $A ~ A(1)[1] ~ A(1)[2] ~ (S(2)[1])(1)[1] ~ (S(2)[2])(1)[1]$ \\
& & & &  $A ~ a ~ b ~ (A(1)[1])(1)[1] ~ (A(1)[2])(1)[1]$ \\
$L_1$ & = & $A ~ a ~ b ~ a ~ b$ & $\longleftarrow$ & $A ~ a ~ b ~ a(1)[1] ~ b(1)[1]$ \\
$L_2$ & $=$ & $a ~ b ~ a ~ b ~ a ~ b$. & &
\end{tabular}
\end{center}

Our new measure $\nu$ is then reachable, strictly better than $b$ and incomparable with $\delta$. It is likely, however, that computing $\nu$ (i.e., finding the smallest NU-system generating a given string $w[1,n]$) is NP-hard.

NU-systems easily allow us concatenating and composing automorphisms. 

\begin{theorem}
Let $N_1=(V_1,R_1,S_1,\tau_1,d_1,n_1)$ and $N_2=(V_2,R_2,S_2,\tau_2,d_2,n_2)$ be NU-systems generating $w_1$ and $w_2$, respectively. Then there are NU-systems of size $O(size(N_1)+size(N_2))$ that generate $w_1 \cdot w_2$ and the {\em composition} of $w_1$ and $w_2$, which is the string generated by $N_2$ with axiom $w_1$, $(V_2,R_2,w_1,\tau_2,d_2,n_2)$.
\end{theorem}
\begin{proof}
Let $V_1' = \{ a_1 \,|\, a \in V_1\}$ and $V_2' = \{ a_2 \,|\, a \in V_2\}$ be disjoint copies of $V_1$ and $V_2$, respectively, and let $R_i'$, and $S_i'$ be variants that operate on $V_i'$ instead of $V_i$.
We build a NU-system $N=(V,R,S,\tau,1,n_1+n_2)$ for $w_1 \cdot w_2$, where $V= V_1' \cup V_2' \cup V_1 \cup V_2 \cup \{ Z_1,Z_2 \}$, where $Z_1$ and $Z_2$ are new symbols. Let $R = R_1' \cup R_2' \cup \{ Z_1 \rightarrow S_1', Z_2 \rightarrow S_2'\}$, plus the rules $a \rightarrow a$ for $a \in V_1 \cup V_2$. The axiom is then $S = Z_1(d_1) \cdot Z_2(d_2)$. Finally, the mapping on $V'_i$ is $\tau(a_i)=\tau_i(a)$, and $\tau(a)=a$ for $a \in V_1 \cup V_2$. It is easy to see that $N$ generates $w_1 \cdot w_2$.

To generate the composition, $V_2$ should contain the image of $V_1$ by $\tau_1$, but still $V_1'$ is disjoint from $V_2'$. The axiom is $Z_1(d_1)$. The mapping on $V_1'$ is $\tau(a_1)=\tau_1(a)_2$. On $V_2'$ we use $\tau(a_2)=\tau_2(a)$. On $V_1 \cup V_2$, we use $\tau(a)=a$. The depth is $1+d_2$. \qed
\end{proof}

The theorem allows a family $\mathcal{F}$ to have $\nu = O(1)$, by finding a finite collection of families generated by fixed non-erasing automorphisms, and then joining them using a finite number of set unions, concatenations and morphism compositions. 

\section{Future work}

We leave a number of open questions. We know $\nu = O(m) = O(b) = O(\delta\log(n/\delta))$, but it is unknown if $\nu = O(\gamma)$; if so, then $\gamma$ would be reachable.
We know $\ell = O(g)$, but it is unknown if $\ell = O(c)$; we suspect it is not, but in general we lack mechanisms to prove lower bounds on $\ell$ or $\nu$. We also know $m = O(b)$, but not if it can be strictly better.
%
We also do not know if these measures are monotone, and if they are actually NP-hard to compute (they are likely so).

We could prove that $g = O(\ell \log n)$, and thus the lower bound $\ell = \Omega(g / \log n)$, if every L-system could be made expanding, but this is also unknown. This, for example, would prove that the stretch $\ell = O(\delta/\log n)$ we found for a family of strings is the maximum possible.


\section{Conclusions}

Extending the study of repetitiveness measures, from parsing-based to morphism-based mechanisms, opens a wide number of possibilities for the study of repetitiveness. There is already a lot of theory behind morphisms, waiting to be exploited on the quest for a golden measure of repetitiveness.

We first generalized composition systems to macro systems, showing that a restriction of them, called internal macro systems, are equivalent to bidirectional macro schemes, the lowest reachable measure of repetitiveness considered in the literature. It is not yet known if general macro systems are more powerful. 

We then showed how morphisms, and measures based on mechanisms capturing that concept called L-systems (and variations), can be strictly better than $\delta$ for some string families, thereby questioning the validity of $\delta$ as a lower bound for reachable repetitiveness measures. L-systems are never larger than context-free grammars, but probably not always as small as composition systems.

Finally, we proposed a novel mechanism of compression aimed at unifying parsing and morphisms as repetitiveness sources, called NU-systems, which builds on macro systems and L-systems. NU-systems can combine copy-paste, recurrences, and concatenations and compositions of morphisms. The size $\nu$ of the smallest NU-system generating a string is a relevant measure of repetitiveness because it is reachable, computable, always in $O(b)$ and sometimes $o(\delta)$.


A simple lower bound capturing the idea of recurrence on a string, and lower bounding $\ell$, just like $\delta$ captures the idea of copy-paste and strictly lower bounds $b$, would be of great interest when studying morphism-based measures. For infinite strings, there exist concepts like {\em recurrence constant} and {\em appearance constant} \cite{Shallit2003}, but an adaptation, or another definition, is needed for finite strings. 
Besides, like Lindenmayer systems, NU-systems could be used to model other repetitive structures beyond strings that appear in biology, like the growth of plants and fractals. In this sense, they can be compared with tree grammars; the relation between NU-systems and TSLPs \cite{Loh15}, for example, deserves further study.


%
%
%
\bibliographystyle{splncs04}
\bibliography{bibliography}

\end{document}